\documentclass[12pt,reqno]{amsart}

\usepackage{epsfig}
\usepackage{amsmath, amsthm, amsfonts}
\usepackage{graphicx}

\numberwithin{equation}{section}

\newcommand{\R}{{\mathbb R}}

\newcommand{\N}{{\mathbb N}}

\newcommand{\Z}{{\mathbb Z}}

\renewcommand{\d}{\partial}

\newcommand{\sign}{{\operatorname{sgn}\,}}

\newcommand{\al}{\alpha}
\newcommand{\be}{\beta}
\newcommand{\ga}{\gamma}

\newcommand{\la}{\lambda}
\newcommand{\ep}{\varepsilon}
\newcommand{\de}{\delta}
\newcommand{\De}{\Delta}

\newcommand{\sg}{\sigma}

\newcommand{\om}{\omega}
\newcommand{\Om}{\Omega}
\newcommand{\z}{\zeta}
\newcommand{\Jth}{\vartheta}

\newcommand{\ocal}{{\mathcal O}}
\newtheorem{theo}{{\sc \bf Theorem}}[section]

\newtheorem{prop}[theo]{{\sc \bf Proposition}}

\begin{document}

\title[Six-vertex model with half-turn symmetry]
{Domain wall six-vertex model with half-turn symmetry}

\author{Pavel Bleher}
\address{Department of Mathematical Sciences,
Indiana University-Purdue University Indianapolis,
402 N. Blackford St., Indianapolis, IN 46202, U.S.A.}
\email{bleher@math.iupui.edu}

\author{Karl Liechty}
\address{Department of Mathematical Sciences,
DePaul University,
Chicago, IL 60614, U.S.A.}
\email{kliechty@depaul.edu}

\thanks{This work was performed during the program {\it Statistical mechanics and combinatorics} at the Simons Center for Geometry and Physics in March 2016. The authors are grateful for the hospitality of the Simons Center. KL would like to thank T. Kyle Petersen for helpful discussions. % involving Proposition \ref{conlaws}.
 PB is supported in part
by the National Science Foundation (NSF) Grants DMS-1265172 and DMS-1565602. 
KL is supported by a grant from the Simons Foundation (\#357872)}

\date{\today}

\begin{abstract}
We obtain asymptotic formulas for the partition function of the six-vertex model with domain wall boundary conditions and half-turn symmetry in each of the phase regions. The proof is based on the Izergin--Korepin--Kuperberg determinantal formula for the partition function, 
its reduction to orthogonal polynomials, and on an asymptotic analysis of the orthogonal polynomials under
consideration in the framework of the Riemann--Hilbert approach. 
\end{abstract}


\maketitle

\section{Introduction}
%This paper continues a program in which the partition function for the six-vertex model with domain wall boundary conditions (DWBC) is rigorously analyzed in the thermodynamic limit.

 The six-vertex model with domain wall boundary conditions (DWBC) was originally introduced by Korepin \cite{Kor82}, who derived some recurrences for the model which were subsequently solved by Izergin \cite{Ize87}, giving an exact formula for the partition function of determinantal type known as the Izergin--Korepin formula. The original Izergin--Korepin formula gives the partition function for a certain {\it inhomogeneous} six-vertex model, in which the weights in the model are site-dependent. The homogeneous version, in which weights depend on vertex-type but are independent of site, is obtained from the inhomogeneous Izergin--Korepin formula by taking a limit as all inhomogeneity parameters coincide. We remark here that for certain values of the inhomogeneity parameters, there is a different determinantal formula for the DWBC partition function which is useful for computation of correlations in some cases \cite{K-ZJ2}, see also \cite{dG-K}. In this paper we focus on the homogeneous six-vertex model with DWBC, and therefore use the Izergin--Korepin formula which has a homogeneous limit.
 
 The states of the six-vertex model with DWBC correspond bijectively with {\it alternating sign matrices}, and the Izergin--Korepin formula was used by Kuperberg \cite{Kup96} to prove the exact enumeration of alternating sign matrices which had been conjectured by Mills, Robbins, and Rumsey \cite{MRR82} and proven by other means in \cite{Zei96}. Later the Izergin--Korepin formula was used by Korepin and Zinn-Justin \cite{K-ZJ} to derive the free energy of the six-vertex model with DWBC. This analysis was continued in \cite{ZJ}, in which the similarity between the Izergin--Korepin formula and random matrix partition functions was noted, thus allowing for an expression of the DWBC partition function in terms of orthogonal polynomials. For certain weights in the model these orthogonal polynomials are classical, and these special cases were studied in a series of papers by Colomo and Pronko \cite{CP1, CP2, CP3, CP4}. Outside of these special weights, the relevant orthogonal polynomials are not classical but may be analyzed asymptotically using the Riemann--Hilbert method. This approach has been employed by the current authors and their collaborators to obtain exact asymptotic formulas for the partition function of the six-vertex model with DWBC as well as partial domain wall boundary conditions (pDWBC) \cite{BF06,BL09-1, BL09-2, BL10, BlBo12, BlBo14, BL14, BL15}. This paper continues that program, extending the asymptotic analysis of the DWBC partition function employed in previous works to a six-vertex model with DWBC which is further constrained by a rotational symmetry.
 
In the work \cite{Kup02} Kuperberg considered various symmetry classes of alternating sign matrices, and was able to give exact enumerations for several such symmetry classes. Each of the symmetry classes studied in that paper corresponds to a different boundary condition for the six-vertex model, and Kuperberg was able to give exact formulas for the partition functions of the corresponding six-vertex models. Each of these formulas is reminiscent of the Izergin--Korepin formula for the DWBC partition function, and involves either determinants or Pfaffians. In this paper we consider one of those symmetry classes, the {\it half-turn} symmetry, in which the DWBC six-vertex states are forced to have a  rotational symmetry of $180^\circ$. 

%For certain non-homogeneous versions of the six-vertex model with DWBC, a different determinantal formula was found

\subsection{Definition of the model}
The half-turn-invariant six-vertex model with DWBC is realized on a rectangular lattice of size $(2n) \times (2n)$ for any $n\in \N$. The states of the model are realized by placing arrows on the edges of the graph obeying the {\it ice rule}: at each vertex there are exactly two arrows pointing in and two arrows pointing out. The arrows on the left and right boundaries are fixed to point out of the lattice, and the arrows on the top and bottom boundaries are fixed to point in. We place the additional constraint that arrow configurations must be invariant with respect to rotation by $180^\circ$, see Figure \ref{example1}.

According to the ice rule, there are exactly six types of configuration at each vertex, and we label them with the numbers $1, 2, \dots,6$ as in Figure \ref{arrows}. The Gibbs measure is defined by assigning a weight $w_k$, $k=1,\dots,6$ to each vertex-type. The weight of an arrow configuration $\sg$ is then defined as
\begin{equation*}\label{int1}
w(\sg)=\prod_{x\in V_{2n}} w_{t(x; \sg)}=\prod_{i=1}^6 w_i^{N_i(\sg)}\,,
\end{equation*}
where $V_{2n}$ is the set of vertices in the lattice, $t(x; \sg)$ is the type of vertex at the vertex $x\in V_{2n}$ in the configuration $\sg$, and $N_i(\sg)$ is the number of vertices of type $i$ in the configuration $\sg$. The Gibbs measure on states is then defined as 
\begin{equation*}\label{int2}
\mu(\sg)= \frac{w(\sg)}{Z_{2n}^{\rm HT}}\,, \qquad Z^{\rm HT}_{2n}\equiv Z_{2n}^{\rm HT}(w_1, w_2, w_3, w_4, w_5, w_6) =\sum_{\sg} w(\sg),
\end{equation*}
where $Z_{2n}^{\rm HT}$ is the {\it partition function}, and the sum is over all configurations obeying both DWBC and the half-turn symmetry.

%%%%%%%%%%%%%%%%
\begin{figure}
\begin{center}\scalebox{0.45}{\includegraphics{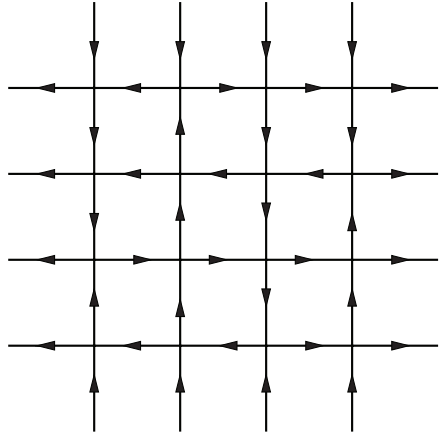}}\end{center}
\caption{An example of the arrow configuration satisfying domain wall boundary conditions and half-turn symmetry on the $4\times 4$ lattice.}
\label{example1}
\end{figure}
%%%%%%%%%%%%%%%%%%%

%%%%%%%%%%%%%%%%
\begin{figure}
\begin{center}\scalebox{0.55}{\includegraphics{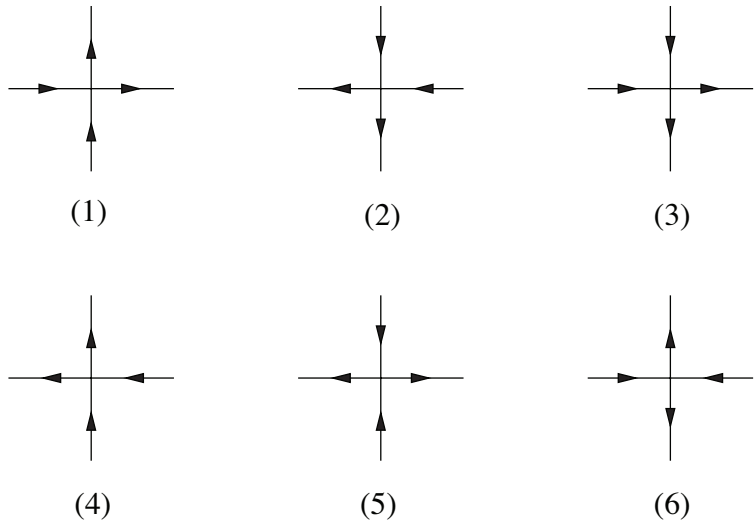}}\end{center}
\caption{The six types of vertices allowed under the ice-rule.}
\label{arrows}
\end{figure}
%%%%%%%%%%%%%%%%%%

A priori there are six parameters in this model: the weights $w_i$. But in fact the boundary conditions impose some conservation laws which allow us to reduce the number of parameters to 3. Namely, any six-vertex configuration $\sg$ on a $(2n) \times (2n)$ lattice satisfying DWBC satisfies the following equations:
\begin{equation}\label{int3}
\begin{aligned}
N_1(\sg)+N_2(\sg)+N_3(\sg)+N_4(\sg)+N_5(\sg)+N_6(\sg)&=4n^2, \\
N_5(\sg)-N_6(\sg)&=2n, \\
N_1(\sg)&=N_2(\sg), \\
N_4(\sg)&=N_3(\sg), \\
\end{aligned}
\end{equation}
see, e.g., \cite{BL09-1, BL14}. Setting 
\begin{equation*}
a=\sqrt{w_1w_2}, \quad b=\sqrt{w_3 w_4}, \quad c=\sqrt{w_5 w_6},
\end{equation*}
the equations \eqref{int3} imply the relation between partition functions,
\begin{equation*}\label{int12}
\begin{aligned}
Z_{2n}^{\rm HT}(w_1, w_2, w_3, w_4, w_5, w_6)
&=\left(\frac{w_5  }{w_6  }\right)^{n}\\\
&\times Z_{2n}^{\rm HT}(a ,a ,b ,b , c, c),
\end{aligned}
\end{equation*}
and between Gibbs measures,
\begin{equation*}\label{int13}
\mu(\sg; w_1, w_2, w_3, w_4, w_5, w_6)=\mu(\sg;a e^{-\eta},a e^{\eta},b e^{-\eta},b e^{\eta}, c, c).
\end{equation*}
Furthermore, using the first equation of \eqref{int3}, we have
\begin{equation*}\label{con5}
\begin{aligned}
Z_{2n}^{\rm HT}(a,a ,b ,b , c, c)&=c^{4n^2}Z_{2n}^{\rm HT}\left(\frac{a }{c}, \frac{a}{c},
 \frac{b }{c}, \frac{b }{c}, 1, 1\right), \\ 
\mu(\sg;a ,a ,b ,b , c, c)
&=\mu\left(\sg; \frac{a }{c}, \frac{a }{c},
 \frac{b }{c}, \frac{b }{c}, 1, 1\right),
\end{aligned}
\end{equation*}
and so the model reduces to the two parameters, $\frac{a}{c}$ and $\frac{b}{c}$.

\subsection{Main results}
The main result of this paper is an asymptotic expansion of the partition function $Z_{2n}^{\rm HT}$ as $n\to\infty$. For finite $n$, a determinantal formula for this partition function was given by Kuperberg, in the spirit of the Izergin--Korepin formua for the DWBC partition function. We refer to this formula as the Izergin--Korepin--Kuperberg formula. In order to state the results, it is convenient to parametrize the weights $a, b$, and $c$ in slightly different ways in the different regions of the phase diagram. 
The phase diagram of the DWBC six-vertex model consists of three phase regions which are described nicely in terms of the parameter 
\[
\De = \frac{a^2+b^2-c^2}{2ab}.
\] If $\De >1$ then we are in the {\it ferroelectric phase}; if $\De<-1$ then we are in the {\it anti-ferroelectric phase}; and if $-1<\De<1$, then we are in the {\it disordered phase}. 

For each of the phase regions, it is convenient to parametrize the weights in slightly different ways. For the ferroelectric phase,
\begin{equation}\label{pf4}
a=\sinh(t-\ga), \quad
b=\sinh(t+\ga), \quad
c=\sinh(2|\ga|), \quad
0<|\ga|<t;
\end{equation}
for the anti-ferroelectric phase,
\begin{equation}\label{pf5}
a=\sinh(\ga-t), \quad
b=\sinh(\ga+t), \quad
c=\sinh(2\ga), \quad
|t|<\ga;
\end{equation}
and for the disordered phase
\begin{equation}\label{pf6}
a=\sin(\ga-t), \quad
b=\sin(\ga+t), \quad
c=\sin(2\ga), \quad
|t|<\ga<\frac{\pi}{2}\,.
\end{equation}

It turns out that the phase diagram for the half-turn-invariant six-vertex model with DWBC differs from that of the usual six-vertex model with DWBC by the quadratic change of weights $a\mapsto \sqrt{a}$, $b\mapsto \sqrt{b}$, $c\mapsto \sqrt{c}$. It is therefore convenient to consider the half-turn invariant model with weights $\sqrt{a}$, $\sqrt{b}$, and $\sqrt{c}$.
The Izergin--Korepin--Kuperberg formula for the partition function $Z_{2n}^{\rm HT}$ is then described in the following proposition.
\begin{prop}\label{prop:finite_n}
Consider the six-vertex model with DWBC on the $(2n)\times(2n)$ lattice with $180^\circ$ rotational symmetry with weights $w_1= w_2 = \sqrt{a}$, $w_3= w_4 = \sqrt{b}$, $w_5= w_6 = \sqrt{c}$, where the weights are parametrized by \eqref{pf4}, \eqref{pf5}, or \eqref{pf6}, depending on the phase region. Then the partition function $Z_{2n}^{\rm HT}$ is given as
\begin{equation}\label{eq:finite_n_formula}
Z_{2n}^{\rm HT}=\frac{(ab)^{2n^2}}{\prod_{j=0}^{n-1} (j!)^4}\,\tau_n^{\rm DW}\,\tau_n^{\rm HT}\,,
\end{equation}
where
\begin{equation*}\label{in6}
\tau_n^{\rm DW}=\det\left( \phi^{(j+k-2)}(t)\right)_{j,k=1}^n,\quad \phi(t)=\frac{c}{ab}\,,
\end{equation*}
and
\begin{equation}\label{in7}
\tau_n^{\rm HT}=\det\left(\psi^{(j+k-2)}(t)\right)_{j,k=1}^n,\quad \psi(t)=\frac{1}{a}
+\frac{1}{b}\,.
\end{equation}
Here $\phi^{(j+k-2)}(t)$ and $\psi^{(j+k-2)}(t)$ refer to the $(j+k-2)$th derivative with respect to $t$, and the dependence of the functions $\phi(t)$ and $\psi(t)$ on $t$ comes from the dependence of the parameters $a$ and $b$ on $t$ in parameterizations \eqref{pf4}--\eqref{pf6}.
\end{prop}
This proposition follows from the result of Kuperberg for the inhomogeneous six-vertex model with half-turn boundary conditions (HTBC) presented in \cite[Theorem 10]{Kup02}. The formula presented in that paper concerns the six-vertex model on a rectangular lattice of size $n\times (2n)$ with additional edges connecting some of the vertices in the top row. Specifically, the vertex in the top row $k$ steps from the left is connected to the one which is $k$ steps from the right, for all $k=1, 2, \dots, n$, see Figure \ref{arrow_conf}. The states of the model are again realized by placing arrows on the edges of the graph obeying the ice rule, subject to the following boundary conditions. The arrows on the left and right boundaries are fixed to point out of the lattice, and the arrows on the bottom boundary are fixed to point in. The arrows on the top boundary are free, up to the constraint imposed by the connection of the $k$-th and $(2n-k)$-th vertices from the left.
%%%%%%%%%%%%%%%%
\begin{figure}
\begin{center}\scalebox{0.3}{\includegraphics{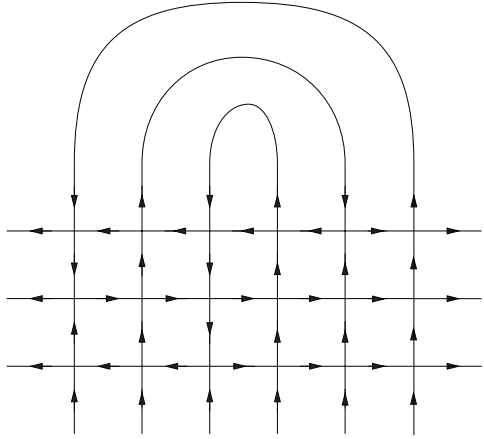}}\end{center}
\caption{An example of the arrow configuration with half-turn boundary conditions on the $3\times 6$ lattice.}
\label{arrow_conf}
\end{figure}
%%%%%%%%%%%%%%%%%%
Denote the partition function for this model as $Z_{n, 2n}^{\rm HTBC}$.
This model is related to the six-vertex model with domain wall boundary conditions on a lattice of size $(2n) \times (2n)$ with a rotational symmetry in a straightforward way. If one were to make two copies a HTBC state, rotate one of the copies $180^\circ$, and glue the two copies together along the top boundary, the result would be a six-vertex state on the $(2n) \times (2n)$ lattice with domain wall boundary conditions and $180^\circ$ rotational symmetry. Under rotation by $180^\circ$, vertices of type 1 and 2 are interchanged, vertices of types 3 and 4 are interchanged, and vertices of type 5 and 6 are invariant. It follows that
\begin{equation*}
Z_{n, 2n}^{\rm HTBC}(a ,a, b ,b , c, c) = Z_{2n}^{\rm HT}(\sqrt{a}, \sqrt{a}, \sqrt{b}, \sqrt{b}, \sqrt{c}, \sqrt{c}).
\end{equation*}
In order to obtain \eqref{eq:finite_n_formula}, one needs to take a homogeneous limit of Kuperberg's formula for $Z_{n, 2n}^{\rm HTBC}(a ,a, b ,b , c, c)$. This procedure is outlined in, e.g., \cite[Section 5.5]{BL14}.

Observe that we can write \eqref{eq:finite_n_formula} as
 \begin{equation}\label{in8}
Z_{2n}^{\rm HT} = \left(Z_n^{\rm DW}\right)\left(\widetilde{Z}_n^{\rm HT}\right), 
\end{equation}
where
\begin{equation}\label{in9}
Z_n^{\rm DW} = \frac{(ab)^{n^2}}{\prod_{j=0}^{n-1} j!^2} \tau_n^{\rm DW}, \qquad \widetilde{Z}_n^{\rm HT} = \frac{(ab)^{n^2}}{\prod_{j=0}^{n-1} j!^2} \tau_n^{\rm HT}.
\end{equation}
The factor $Z_n^{\rm DW} $ is exactly the partition function for the six-vertex model on an $n\times n$ lattice with domain wall boundary conditions, and has been evaluated asymptotically as $n\to\infty$ in each of the phase regions a series of papers \cite{BF06, BL09-1, BL10}. 
Thus to evaluate \eqref{eq:finite_n_formula} in the thermodynamic limit, we need only evaluate $\widetilde{Z}_n^{\rm HT}$, or equivalently the determinant $\tau_n^{{\rm HT}}$, asymptotically as $n\to\infty$. 
That analysis comprises the main technical work of this paper, and it leads to the following results for the asymptotics of the half-turn-invariant partition function.

\begin{theo}[Disordered phase]\label{theorem:D}
Let $Z_{2n}^{\rm HT}$ be the partition function for the six-vertex model with DWBC on the lattice of size $(2n)\times (2n)$ and half-turn symmetry, with the weights
\begin{equation}\label{wd}
\begin{aligned}
&w_1=w_2=\sqrt{\sin(\ga-t)}, \quad
w_3=w_4=\sqrt{\sin(\ga+t)}, \\
&w_5=w_6=\sqrt{\sin(2\ga)}, \quad
|t|<\ga<\frac{\pi}{2}\,.
\end{aligned}
\end{equation}
Then as $n\to \infty$, 
\begin{equation*}
Z_{2n}^{\rm HT} =C n^{\kappa} F^{2n^2}(1+\ocal(n^{-1})),
\end{equation*}
where 
\begin{equation*}
F=\frac{\pi \sin(\ga-t)\sin(\ga+t)}{2\ga\cos \frac{\pi t}{2\ga}}, \quad \kappa = \frac{1}{6} - \frac{\ga^2}{3\pi(\pi-2\ga)},
\end{equation*}
and the constant $C$ is of the form
\begin{equation}\label{constant:D}
 C=\left[\cos\left(\frac{\pi t}{2\ga}\right)\right]^{\kappa} D(\ga),
\end{equation}
where $D(\ga)$ does not depend on $t$.
\end{theo}

 \begin{theo}[Antiferroelectric phase]\label{theorem:AF}
Let $Z_{2n}^{\rm HT}$ be the partition function for the six-vertex model with DWBC on the lattice of size $(2n)\times (2n)$ and half-turn symmetry, with the weights
\begin{equation}\label{wd}
\begin{aligned}
&w_1=w_2=\sqrt{\sinh(\ga-t)}, \quad
w_3=w_4=\sqrt{\sinh(\ga+t)}, \\
&w_5=w_6=\sqrt{\sinh(2\ga)}, \quad
|t|<\ga\,.
\end{aligned}
\end{equation}
Then as $n\to \infty$,
\begin{equation*}
Z_{2n}^{\rm HT} =C \Jth_3(n\om)\Jth_4(n\om)  F^{2n^2}(1+\ocal(n^{-1})),
\end{equation*}
where 
\begin{equation*}
F=\frac{\pi \sinh(\ga-t)\sinh(\ga+t) \Jth_1'(0)}{\ga\Jth_1(\om)}, \quad \om = \frac{\pi}{2}(1+\z), \quad \z=\frac{t}{\ga},
\end{equation*}
$\Jth_1$, $\Jth_3$, and $\Jth_4$ are the Jacobi theta functions (see \eqref{main8_af} for their definitions) with elliptic nome $q=e^{-\frac{\pi^2}{2\ga}},$
and the constant $C$ is independent of both $n$ and $t$.
\end{theo} 

%For the definitions of the Jacobi elliptic functions, see equation \eqref{main8_af}.

 \begin{theo}[Ferroelectric phase]\label{theorem:F}
Let $Z_{2n}^{\rm HT}$ be the partition function for the six-vertex model with DWBC on the lattice of size $(2n)\times (2n)$ and half-turn symmetry, with the weights
\begin{equation}\label{wd}
\begin{aligned}
&w_1=w_2=\sqrt{\sinh(t-\ga)}, \quad
w_3=w_4=\sqrt{\sinh(t+\ga)}, \\
&w_5=w_6=\sqrt{\sinh(2\ga)}, \quad
0<\ga<t\,.
\end{aligned}
\end{equation}
Then for any $\ep>0$ as $n\to \infty$, 
\begin{equation*}
Z_{2n}^{\rm HT} =C G^{2n} F^{2n^2}(1+\ocal(e^{-n^{1-\ep}})),
\end{equation*}
where 
\begin{equation}\label{F_constants}
F=\sinh(t+\ga), \quad G= e^{\ga-t}, \quad C=(1+e^{-4\ga})(1-e^{-4\ga}).
\end{equation}
\end{theo}

The proofs of the above theorems follow very closely the proofs of the asymptotic expansion of $Z_n^{\rm DW}$, in which the Hankel determinant $\tau_n^{\rm DW}$ is expressed in terms of a system of orthogonal polynomials which is then evaluated asymptotically as the degree of the polynomial becomes large. Any information about the constant term in the asymptotic expansion is gleaned from the fact that $\tau_n^{\rm DW}$ satisfies the {\it Toda equation},
\begin{equation}\label{toda}
\frac{\d^2}{\d t^2} \log \tau_n = \frac{\tau_{n+1}\tau_{n-1}}{\tau_n^2}.
\end{equation}

In each of the phase regions, the asymptotic analysis of $\tau_n^{\rm HT}$ is very similar. It also satisfies the Toda equation, so \eqref{toda} holds for either $\tau_n=\tau_n^{\rm DW}$ or $\tau_n=\tau_n^{\rm HT}$.  Also $\tau_n^{\rm HT}$ may be expressed in terms of a system of orthogonal polynomials, and it turns out these orthogonal polynomials are very similar to the ones for $\tau_n^{\rm DW}$. Indeed, asymptotically they differ only in subleading terms.  

\subsection{Plan for the rest of the paper}
In Section \ref{laplace}, the symbol $\psi(t)$ appearing in the Hankel determinant $\tau_n^{\rm HT}$ is expressed as a Laplace-type transform of some measure on the real line in each of the phase regions. This is the first step to expressing $\tau_n^{\rm HT}$ in terms of orthogonal polynomials. Then Theorems \ref{theorem:D}, \ref{theorem:AF}, and \ref{theorem:F} are proven in Sections \ref{d-phase}, \ref{af-phase}, and \ref{f-phase}, respectively. In each case, two systems of orthogonal polynomials are discussed: those relevant for $\tau_n^{\rm DW}$ and those for $\tau_n^{\rm HT}$. The asymptotic results for the former system are recalled, and the new asymptotic results for the latter system are stated. Since the analysis is very similar for both systems of orthogonal polynomials, we simply describe the adjustments that must be made to the analysis of the polynomials corresponding to $\tau_n^{\rm DW}$, as presented in \cite{BL14}, in order to apply to the polynomials corresponding to $\tau_n^{\rm HT}$. In each of the sections \ref{d-phase}--\ref{f-phase} we use the same notation for the relevant systems of orthogonal polynomials even though they are different in each section. We trust it will not confuse the reader.

\subsection{Outlook}
As noted earlier in the introduction, Kuperberg \cite{Kup02} has given exact formulas for the partition functions of the six-vertex model with several different boundary conditions/symmetry classes. In \cite[Theorem 10]{Kup02}, we find a list of eight such formulas. This list includes the well-known DWBC partition function as well as the partition function for the U-boundary condition which had been previously found by Tsuchiya \cite{Tsu98}. Tsuchiya's determinantal formula is generalized to the UU-boundary condition, and new formulas are also given for half-turn (HT), quarter-turn (QT), off-diagonal (O), off-diagonal-off-anti-diagonal (OO), and U-off-antidiagonal (UO) boundary conditions. After DWBC and HTBC, it is natural to ask whether the remaining partition functions in this list may be analyzed asymptotically.

For the U- and UU-boundary conditions, the formulas for $Z_n$ are determinantal and involve the Tsuchiya determinant. After taking the homogeneous limit this determinant is not a Hankel determinant, but may be expressed in terms of certain {\it bi}-orthogonal polynomials. Rather than satisfying the Toda equation \eqref{toda}, the Tsuchiya determinant satisfies a two dimensional version of \eqref{toda}. This fact was used in \cite{RK15} to derive the free energy for the U-boundary condition partition function in the disordered phase. A rigorous asymptotic analysis following \cite{BL14} would require rather general machinery for asymptotic analysis of bi-orthogonal polynomials. At the moment such machinery is missing.

For the remaining four partition functions formulated by Kuperberg (QT, O, OO, UO), the formulas involve Pfaffians rather than determinants. For Pfaffians, orthogonal polynomial methods do not apply, although it may be possible to write those formulas in terms of systems of skew-orthogonal polynomials. Again there is not currently any general machinery for asymptotic analysis of skew-orthogonal polynomials, so the methods of this paper and \cite{BL14} do not apply.

\section {The Laplace transform}\label{laplace}
The first step in the asymptotic analysis of $\tau_n^{\rm HT}$ is to write the symbol $\psi(t)= 1/a + 1/b$ as the Laplace transform of some measure on the real line. This representation is different for the different phase regions, and is given in the following proposition.
\begin{prop}\label{lt} 
 In the ferroelectric phase, in which $a$ and $b$ parametrized as \eqref{pf4}, we have that
\begin{equation}\label{lt3a}
\psi(t)=\frac{1}{\sinh(t-\ga)}
+\frac{1}{\sinh(t+\ga)}=2\sum_{k=1}^\infty [e^{-2k (t-\ga)}+e^{-2k(t+\ga)}].
\end{equation}
 In the antiferroelectric phase, in which $a$ and $b$ parametrized as \eqref{pf5}, we have that
 \begin{equation}\label{lt3b}
\psi(t)=\frac{1}{\sinh(\ga-t)}
+\frac{1}{\sinh(\ga+t)}=2\sum_{k=-\infty}^\infty e^{(2k+1)t-|2k+1|\ga}.
\end{equation}
In the disordered phase, in which $a$ and $b$ parametrized as \eqref{pf6}, we have that
\begin{equation}\label{lt3}
\psi(t)=\frac{1}{\sin(\ga-t)}
+\frac{1}{\sin(\ga+t)}=\int_{-\infty}^\infty e^{t\la} m(\la)\,d\la\,,
\end{equation}
where 
\begin{equation}\label{lt4}
m(\la)=\frac{e^{-\ga\la}}{1+e^{-\pi\la}}+\frac{e^{\ga\la}}{1+e^{\pi\la}}\,.
\end{equation}
\end{prop}

The form of these Laplace transform representations of $\psi(t)$ determines the orthogonal polynomials which appear in subsequent asymptotic analysis: in Section \ref{d-phase} we consider continuous orthogonal polynomials on $\R$ with respect to the weight $e^{tx}m(x)$ which appears in the integrand of \eqref{lt3}; in Section \ref{af-phase} we consider the discrete weight $(2k+1)t-|2k+1|\ga$, $k\in \Z$ as in \eqref{lt3b}; and in Section \ref{f-phase} we consider the discrete weight $e^{-2k (t-\ga)}+e^{-2k(t+\ga)}$, $k\in \Z_+$ which appears in \eqref{lt3a}.

\begin{proof}[Proof of Proposition \ref{lt}] Equations \eqref{lt3a} and \eqref{lt3b} follow in a straightforward way using geometric series.
For \eqref{lt3a} we have
\begin{equation*}
\begin{aligned}
\frac{1}{\sinh(t-\ga)}+\frac{1}{\sinh(t+\ga)}&=\frac{2e^{-(t-\ga)}}{1-e^{-2(t-\ga)}}+\frac{2e^{-(t+\ga)}}{1-e^{-2(t+\ga)}} \\
&=2e^{-(t-\ga)}\sum_{k=0}^\infty e^{-2k(t-\ga)}+2e^{-(t+\ga)}\sum_{k=0}^\infty e^{-2k(t+\ga)} \\
%&=4\sum_{k=0}^\infty e^{-(2k+1)t}\cosh((2k+1)\ga) \\
&=2\sum_{k=0}^\infty [e^{-(2k+1) (t-\ga)}+e^{-(2k+1) (t+\ga)}] \\
&=2e^{-(t-\ga)}\sum_{k=0}^\infty [e^{-2k (t-\ga)}+e^{-2k(t+\ga)-2\ga}] \\
&=2\sum_{k=1}^\infty [e^{-2k (t-\ga)}+e^{-2k(t+\ga)}],
\end{aligned}
\end{equation*}
where we have used the restriction $|\ga|<t$ to ensure convergence of the series. The proof of \eqref{lt3b} is very similar, where we use $|t|<\ga$ to ensure convergence, and we omit it here. We are left only to prove \eqref{lt3} and \eqref{lt4}.

Let $t=i\tau$. Then \eqref{lt3} reads
\begin{equation*}\label{lt5}
\psi(i\tau)=\frac{1}{\sin(\ga-i\tau)}
+\frac{1}{\sin(\ga+i\tau)}=\int_{-\infty}^\infty e^{i\tau\la} m(\la)\,d\la
=\int_{-\infty}^\infty e^{i\tau\la} [m_-(\la)+m_+(\la)]\,d\la\,,
\end{equation*}
hence, taking the inverse Fourier transform,
\begin{equation*}\label{lt6}
m_+(\la)=\frac{1}{2\pi} 
\int_{-\infty}^\infty \frac{1}{\sin(\ga+i\tau)}\,e^{-i\tau\la}d\tau\,.
\end{equation*}
Shifting the contour of integration up by $i\pi$ and evaluating the residue at $\tau=i\ga$, we obtain the equation,
\begin{equation*}\label{lt7}
m_+(\la)=-m_+(\la)e^{\pi\la}+e^{\ga\la},
\end{equation*}
hence
\begin{equation*}\label{lt8}
m_+(\la)=\frac{e^{\ga\la}}{1+e^{\pi\la}}\,.
\end{equation*}
Similarly, for
\begin{equation*}\label{lt9}
m_-(\la)=\frac{1}{2\pi} 
\int_{-\infty}^\infty \frac{e^{-i\tau\la}}{\sin(\ga-i\tau)}\,d\tau\,.
\end{equation*}
we obtain the equation,
\begin{equation*}\label{lt10}
m_-(\la)=-m_+(\la)e^{-\pi\la}+e^{-\ga\la},
\end{equation*}
hence
\begin{equation*}\label{lt11}
m_-(\la)=\frac{e^{-\ga\la}}{1+e^{-\pi\la}}\,,
\end{equation*}
and \eqref{lt3} follows.
\end{proof}

%In Section \ref{d-phase} we analyze the disordered phase and we prove Theorem \ref{theorem:D}.

\section{Disordered Phase}\label{d-phase}

Let us consider two different systems of orthogonal polynomials on $\R$. They are defined in terms of the two orthogonality weights
\begin{equation}\label{wt1}
w^{\rm DW}(x)=\frac{e^{tx} \sinh\left[\left(\frac{\pi}{2}-\ga\right)x\right]}{\sinh\frac{\pi x}{2}}, \qquad w^{\rm HT}(x)=\frac{e^{tx} \cosh\left[\left(\frac{\pi}{2}-\ga\right)x\right]}{\cosh\frac{\pi x}{2}}.
\end{equation}
We then define two systems of monic orthogonal polynomials $\{P_{k}^{\rm DW}(x)\}_{k=0}^\infty$ and $\{P_{k}^{\rm HT}(x)\}_{k=0}^\infty$ such that
\begin{equation}\label{wt2}
\int_{-\infty}^\infty P_{j}^{\rm DW}(x)P_{k}^{\rm DW}(x)w^{\rm DW}(x)=h_k^{\rm DW}\de_{jk}, \\
\end{equation}
\begin{equation}\label{wt2a}
\int_{-\infty}^\infty P_{j}^{\rm HT}(x)P_{k}^{\rm HT}(x)w^{\rm HT}(x)=h_k^{\rm HT}\de_{jk}, 
\end{equation}
where $\{h_k^{\rm DW}\}_{k=0}^\infty$ and $\{h_k^{\rm HT}\}_{k=0}^\infty$ are two sequences of normalizing constants, and $P_{k}^{\rm DW}$ and $P_{k}^{\rm HT}$ are monic polynomials of degree exactly $k$. Notice that the weight $w^{\rm HT}(x)$ is the same as $e^{tx}m(x)$ where $m(x)$ is defined in \eqref{lt4}. The polynomials orthogonal with respect to $w^{\rm DW}(x)$ are the ones which appear in the study of the partition function of the six-vertex model with domain wall boundary conditions. Due to a general fact about Hankel determinants, see e.g. \cite[equations (4.4.8)--(4.4.14)]{BL14}, the two determinants $\tau_n^{\rm DW}$ and $\tau_n^{\rm HT}$ can be written explicitly in terms of the normalizing constants $h_k^{\rm DW}$ and $h_k^{\rm HT}$, respectively. Indeed we have
\begin{equation*}\label{wt3}
\tau_n^{\rm DW}=\prod_{k=0}^{n-1} h_k^{\rm DW}, \qquad 
\tau_n^{\rm HT}=\prod_{k=0}^{n-1} h_k^{\rm HT}.
\end{equation*}
Thus the partition function \eqref{eq:finite_n_formula} may be written as
\begin{equation}\label{wt4}
Z_{2n}^{\rm HT} = (ab)^{2n^2} \prod_{k=0}^{n-1} \frac{h_k^{\rm DW}}{(k!)^2}\prod_{k=0}^{n-1} \frac{h_k^{\rm HT}}{(k!)^2}.
\end{equation}

In the work \cite{BF06}, see also \cite[Chapter 6]{BL14}, a very precise asymptotic formula is obtained for the sequence of constants $h_k^{\rm DW}$ as $k\to\infty$ using the Riemann--Hilbert method.  The following result is given in \cite[Proposition 6.1.2]{BL14}.
\begin{prop}\label{dis_DW_hk} 
As $k\to \infty$, the normalizing constants $h_k^{\rm DW}$ satisfy
\begin{equation}\label{dpmr4a}
\frac{h_k^{\rm DW}}{(k!)^2}=G^{2k+1} \left(1+\frac{\kappa^{\rm DW}}{k}+\ep_k+\mathcal O(k^{-2})\right)\,,
\end{equation}
where
\begin{equation}\label{dpmr4b}
\begin{aligned}
G=\frac{\pi}{2\ga \cos\left(\frac{\pi\z}{2}\right)}\,, \quad \z=\frac{t}{\ga}, \quad 
\kappa^{\rm DW}=\frac{1}{12}-\frac{2\ga^2}{3\pi(\pi-2\ga)}\,,
\end{aligned}
\end{equation}
and
\begin{equation}\label{dpmr4c}
\ep_k=\left[\cos(k\om)+\tan\left(\frac{\pi\z}{2}\right)\sin(k\om)\right] \sum_{j: \ \kappa_j < 2} c_j k^{-\kappa_j}\,.
\end{equation}
Here 
\begin{equation}\label{dpmr4d}
\om=-\pi(1+\z),\qquad \kappa_j=1+\frac{2j}{\frac{\pi}{2\ga}-1}>1\,,\qquad j=1,2,\ldots,
\end{equation}
and the numbers $c_j$ are equal to
\begin{equation}\label{dpmr4f}
c_j=d_j\sin\left(\frac{\pi j}{1-\frac{2\ga}{\pi}}\right),\qquad j=1,2,\ldots,
\end{equation}
where $d_j\not=0$.
\end{prop}

In light of equation \eqref{wt4}, this proposition is a key ingredient in the proof of Theorem \eqref{theorem:D}. We also need an asymptotic formula for the sequence $h_k^{\rm HT}$. Before presenting that result, we briefly comment on the two systems of orthogonal polynomials \eqref{wt2} and \eqref{wt2a}. The weights \eqref{wt1} are identical 
except that $\cosh$ and $\sinh$ are interchanged. It may not be surprising to the reader then that the asymptotic analysis is very similar for both systems of orthogonal polynomials.   Consider the rescaled weights
\begin{equation*}\label{rw1_d}
\begin{aligned}
w_n^{\rm DW}(x)=w^{\rm DW}\left(\frac{nx}{\ga}\right)&=
\frac{e^{n\z x}
\sinh \left[n\left(\frac{\pi}{2\ga}-1\right)x\right]}
{\sinh \frac{n\pi x}{2\ga}}, \\
w_n^{\rm HT}(x)=w^{\rm HT}\left(\frac{nx}{\ga}\right)&=
\frac{e^{n\z x}
\cosh \left[n\left(\frac{\pi}{2\ga}-1\right)x\right]}
{\cosh \frac{n\pi x}{2\ga}}, \quad \z=\frac{t}{\ga},
\end{aligned}
\end{equation*}
which may be written as
\begin{equation*}\label{rw1_b}
w_n^{\rm DW}(x) =e^{-nV_n^{\rm DW}(x)}, \qquad w_n^{\rm HT}(x) =e^{-nV_n^{\rm HT}(x)},
\end{equation*}
where 
\begin{equation*} \label{rw5_d}
V_{n}^{\rm DW}(x)=-\z x-\frac{1}{n}\,\log\frac{
\sinh \left[n\left(\frac{\pi}{2\ga}-1\right)x\right]}
{\sinh \frac{n\pi x}{2\ga}}, \quad 
V_{n}^{\rm HT}(x)=-\z x-\frac{1}{n}\,\log\frac{
\cosh \left[n\left(\frac{\pi}{2\ga}-1\right)x\right]}
{\cosh \frac{n\pi x}{2\ga}}\,.
\end{equation*}
This rescaling of the weight is the first step in the Riemann--Hilbert analysis of the orthogonal polynomials \eqref{wt2} done in \cite{BF06, BL14}. Notice that as $n\to \infty$, both $V_{n}^{\rm DW}(x)$ and $V_{n}^{\rm HT}(x)$ approach the same limit:
\begin{equation*} \label{rw6a_d}
 V_n^{\rm DW}(x)\to V(x), \qquad  V_n^{\rm HT}(x)\to V(x) \equiv
-\zeta x+|x|.
\end{equation*}
This indicates that the orthogonal polynomials \eqref{wt2a} will differ from the ones \eqref{wt2} only in subleading terms as $k\to \infty$. Indeed, it requires only minor adjustments of the Riemann--Hilbert analysis presented in \cite[Chapter 6]{BL14} to prove the following theorem.
\begin{theo}\label{dis_HT_hk}
As $k\to \infty$, the normalizing constants $h_k^{\rm HT}$ satisfy
\begin{equation}\label{dpmr4a_HT}
\frac{h_k^{\rm HT}}{(k!)^2}=G^{2k+1} \left(1+\frac{\kappa^{\rm HT}}{k}+\ep_k+\mathcal O(k^{-2})\right),
\end{equation}
where $G$ is as in \eqref{dpmr4b},
\begin{equation}\label{dpmr4b_HT}
\begin{aligned}
\kappa^{\rm HT}=\frac{1}{12}+\frac{\ga^2}{3\pi(\pi-2\ga)},
\end{aligned}
\end{equation}
 and $\ep_k$ is as in \eqref{dpmr4c}--\eqref{dpmr4f} except that the numbers $c_j$ are given as
 \begin{equation*}\label{dpmr4f_HT}
c_j=d_j\cos\left(\frac{\pi (j-1/2)}{1-\frac{2\ga}{\pi}}\right),\qquad j=1,2,\ldots,
\end{equation*}
where $d_j\not=0$, and the exponents $\kappa_j$ are
 \begin{equation*}\label{kappajHT}
 \kappa_j = 1+\frac{2j-1}{\frac{\pi}{2\ga}-1}>1.
 \end{equation*}
 
\end{theo}

Combining Theorem \ref{dis_HT_hk}, Proposition \ref{dis_DW_hk}, and equation \eqref{wt4} we immediately obtain that for some $\ep>0$, 
\begin{equation*}
Z_{2n}^{\rm HT} =C n^{\kappa} F^{2n^2}(1+\ocal(n^{-\ep})),
\end{equation*}
where
\begin{equation*}
F=\frac{\pi ab}{2\ga\cos \frac{\pi t}{2\ga}}, \quad \kappa =\kappa^{\rm HT}+\kappa^{\rm DW} = \frac{1}{6} - \frac{\ga^2}{3\pi(\pi-2\ga)},
\end{equation*}
and $C$ is an unknown constant. To complete the proof of Theorem \ref{theorem:D}, we must improve the error to $\ocal(n^{-1})$ and prove that the constant $C$ is of the form \eqref{constant:D}. The fact that the error is in fact $\ocal(n^{-1})$ follows from the oscillatory nature of the constants $\ep_k$ which appear in \eqref{dpmr4a} and \eqref{dpmr4a_HT}. For a proof that $\sum_{k=1}^{n-1} \ep_k = {\rm const.} + \ocal(n^{-1})$, see \cite[proof of Theorem 6.1.3]{BL14}. The proof that the constant term $C$ is of the form \eqref{constant:D} relies on the fact that the determinants $\tau_n^{\rm DW}$ and $\tau_n^{\rm HT}$ both satisfy the Toda equation \eqref{toda}.
 Then the dependence of the constant $C$ on the parameter $t$ can be obtained exactly as in \cite[Section 6.8]{BL14}. This completes the proof of Theorem \ref{theorem:D}, given the result of Theorem \ref{dis_HT_hk}. The proof of this theorem is outlined below.

\begin{proof}[Proof of Theorem \ref{dis_HT_hk}]
The complete Riemann--Hilbert analysis necessary to prove Theorem \eqref{dis_HT_hk} is quite involved and we do not repeat it here since it is nearly identical to the analysis of \cite{BF06}, \cite[Chapter 6]{BL14}. Instead  we simply list the relatively few changes which must be made to the analysis of \cite[Chapter 6]{BL14} to deal with the weight $w^{\rm HT}(x)$ rather than $w^{\rm DW}(x)$. 
The key ingredient in the Riemann--Hilbert analysis is the computation of the equilibrium measure, which is calculated in \cite[Section 6.3]{BL14}. As noted earlier, the two potentials $V_n^{\rm DW}$ and $V_n^{\rm HT}$ have the same limit as $n\to \infty$, and so the calculation of the limiting equilibrium measure for $V_n^{\rm HT}$ is exactly as shown in \cite[Section 6.3.1]{BL14}. The only difference is in the subleading corrections to the equilibrium measure, which are calculated in \cite[Section 6.3.2]{BL14}. In this calculation the primary difference is that $\coth$ must be replaced with $\tanh$ in equations (6.3.24) and (6.3.26) of \cite{BL14}. Then the calculation in equation (6.3.54) becomes\footnote{There is a typo in  \cite[Equation (6.3.54)]{BL14}. All integrals after the first line should be from $0$ to $\infty$, not $-\infty$ to $\infty$. The value of the integral, however, is given correctly.}
\begin{equation*}\label{ap1}
\begin{aligned} 
\int_{-\infty}^\infty xf(x)dx
&=\int_0^\infty x
\left[\frac{\pi}{\ga}\tanh \frac{\pi x}{2\ga}
-\frac{\pi-2\ga}{\ga}
\tanh \frac{(\pi-2\ga)x}{\ga}
-2\,\sign x\right]dx\\
&=2\int_{0}^\infty x
\left[\left(\frac{\pi}{\ga}-2\right)\frac{1}{e^{x(\pi/\ga-2)}+1}-\frac{\pi}{\ga}\,\frac{1}{e^{x\pi/\ga}+1}
\right]dx\\
&=-2\left(\frac{\ga}{\pi}-\frac{\ga}{\pi-2\ga}\right)\int_{0}^\infty\frac{u}{e^u+1}\,du \\
&=-2\left(\frac{\ga}{\pi}-\frac{\ga}{\pi-2\ga}\right)\frac{\pi^2}{12}=\frac{\pi\ga^2}{3(\pi-2\ga)}\,.
\end{aligned}
\end{equation*}
Notice that the final result of this calculation differs from the original one in \cite[Equation (6.3.54)]{BL14} by a factor of $-(1/2)$. As a result, the RHS of equations (6.3.55) and (6.3.56) is changed by a factor of $-(1/2)$, as do the $n^{-2}$-terms in equations (6.3.57), (6.3.58), and (6.3.60).
Thus in Proposition 6.3.2, the terms of the order $n^{-2}$ are changed by a factor of $-(1/2)$. 

The calculations in \cite[Sections 6.3.4, 6.3.5)]{BL14}, in which the density and resolvent of the equilibrium measure are calculated, may be repeated verbatim keeping in mind that the endpoints $\al_n$ and $\be_n$ have been adjusted in the subleading terms. The next section \cite[Section 6.3.6]{BL14} requires some modification. The terms of order $n^{-2}$ in the expansion of $\al_n$ and $\be_n$ are expressed as $\De$ in equation (6.3.90), and so the definition of $\De$ given in (6.3.91) must be changed by a factor of $(-1/2)$, as must the RHS of both equations in (6.3.119). Then the first terms in the RHS of (6.3.128) and (6.3.133) are also multiplied by $(-1/2)$, the calculations in the remainder of the section are adjusted accordingly, and we find that (6.3.138) becomes
\begin{equation*}\label{ap2}
e^{ n l_n}=\frac{(\be-\al)^{2n}}{2^{4n}e^{2n}}\left(1  + \frac{\ga^2}{3n\pi(\pi-2\ga)}+\mathcal O(n^{-2})\right).
\end{equation*}
As a result of the above formula, in equations (6.7.4)-(6.7.6), the term $-\frac{2\ga^2}{3n\pi(\pi-2\ga)}$ is replaced with $\frac{\ga^2}{3n\pi(\pi-2\ga)}$, and so the $n^{-1}$ term in the expansion of $h_n$ becomes 
\begin{equation*}\label{ap3}
\kappa=\frac{1}{12}+\frac{\ga^2}{3n\pi(\pi-2\ga)},
\end{equation*}
which is exactly the $\kappa^{\rm HT}$ given in \eqref{dpmr4b_HT}.

The changes in the Riemann--Hilbert analysis of \cite[Sections 6.4-6]{BL14} are quite minor. In equations (6.5.8), (6.5.13), and (6.5.20), $\sinh$ must be replaced by $\cosh$. Then in equation (6.6.15) the poles are at the points 
\begin{equation*}\label{ap4}
z_j = \frac{i(j-1/2)\pi}{n(\frac{\pi}{2\ga}-1)}.
\end{equation*}
In (6.6.16)\footnote{There is a typo in \cite[Equation (6.6.16)]{BL14}. $z_j$ should be replaced with $y_j$, and the definition of $y_j$ should be made earlier.} and (6.6.18), $\sin$ should be replaced with $\cos$, and $y_j$ should be defined as
\begin{equation*}\label{ap5}
y_j = \frac{(j-1/2)\pi}{\frac{\pi}{2\ga}-1},
\end{equation*}
and in (6.1.26), the numbers $c_j$ are subsequently defined as
\begin{equation*}\label{ap6}
c_j= d_j\cos\left(\frac{\pi(j-1/2)}{1-2\ga/\pi}\right),\quad j=1,2,\dots.
\end{equation*}
\end{proof}

\section{Anti-ferroelectric phase}\label{af-phase}
In the anti-ferroelectric phase $c>a+b$, the weights $a$, $b$, and $c$ are parametrized as in \eqref{pf5}. We consider the two sequences of monic orthogonal polynomials
\begin{equation*}
\sum_{\ell=-\infty}^\infty P^{\rm DW}_j(\ell)P^{\rm DW}_k(\ell) e^{2t\ell - 2\ga|\ell|} = h^{\rm DW}_k \de_{jk}, \quad \sum_{\ell=-\infty}^\infty P^{\rm HT}_j(\ell)P^{\rm HT}_k(\ell) e^{(2\ell+1)t - |2\ell+1|\ga} = h^{\rm HT}_k \de_{jk},
\end{equation*}
where $\{h^{\rm DW}_k\}_{k=0}^\infty$ and $\{h^{\rm HT}_k\}_{k=0}^\infty$ are sequences of normalizing constants. The orthogonal polynomials $P^{\rm DW}_j(\ell)$ appear in the analysis of the six-vertex model with DWBC \cite{BL10}, \cite[Chapter 7]{BL14}, and that partition function is given as
\begin{equation*}
Z_n^{\rm DW} = (2ab)^{n^2}\prod_{k=0}^{n-1} \frac{h_k^{\rm DW}}{(k!)^2}.
\end{equation*}
Using the formula \eqref{lt3b} for the function $\psi(t)$, we find that the Hankel determinant $\tau_n^{\rm HT}$ defined in \eqref{in7} can be expressed in terms of the orthogonal polynomials 
$P^{\rm HT}_j(\ell)$. Indeed we have
\begin{equation*}
\tau_n^{\rm HT}= 2^{n^2} \prod_{k=0}^{n-1} h_k^{\rm HT}.
\end{equation*}
Thus \eqref{in8} and \eqref{in9} can be written as
\begin{equation}\label{eq:Zn_factor}
Z_{2n}^{\rm HT} = (2ab)^{2n^2} \prod_{k=0}^{n-1} \frac{h_k^{\rm DW}}{(k!)^2}\prod_{k=0}^{n-1} \frac{h_k^{\rm HT}}{(k!)^2}.
\end{equation}

The asymptotic formulas for the constants $h_k^{\rm DW}$ and $h_k^{\rm HT}$ involve Jacobi theta functions, so we recall their definitions here. The four Jacobi theta functions are defined as
\begin{equation}\label{main8_af}
\begin{aligned}
\Jth_1(z)\equiv\Jth_1(z;q)&=2\sum_{n=0}^{\infty}(-1)^n q^{(n+\frac{1}{2})^2} \sin\big((2n+1)z\big)\,, \\
\Jth_2(z)\equiv\Jth_2(z;q)&=2\sum_{n=0}^{\infty}q^{(n+\frac{1}{2})^2}\cos\big((2n+1)z\big)\,, \\
\Jth_3(z)\equiv\Jth_3(z;q)&=1+2\sum_{n=1}^{\infty}q^{n^2}\cos(2nz)\,, \\
\Jth_4(z)\equiv\Jth_4(z;q)&=1+2\sum_{n=1}^{\infty}(-1)^n q^{n^2} \cos(2nz)\,, \\
\end{aligned}
\end{equation}
where $q$ is a complex number called the {\it elliptic nome} satisfying $|q|<1$.
Note that the third and fourth theta functions are related by a shift:
\begin{equation*}
\Jth_3(z; q) = \Jth_4\left(z+ \frac{\pi}{2}; q\right).
\end{equation*}
 In what follows we will fix the elliptic nome for all theta functions to be
\begin{equation*}
q = e^{-\frac{\pi^2}{2 \ga}},
\end{equation*}
and suppress the dependence on $q$ from the notation.

In \cite[Proposition 13.1]{BL10}, see also \cite[Proposition 7.3.1]{BL14}, the following asymptotic formula is obtained for the constants $h_k^{\rm DW}$.
\begin{prop}\label{af_DW_hk}
As $k\to\infty$, the normalizing constants $h_k^{\rm DW}$ satisfy
\begin{equation*}\label{af:10}
\frac{h_k^{\rm DW}}{(k!)^2} = G^{2k+1} \frac{\Jth_4((k+1)\om)}{\Jth_4(k\om)} \left(1+\ocal(k^{-2})\right),
\end{equation*}
where
\begin{equation}\label{af:11}
\om=\frac{\pi}{2}(1+\z), \quad G=\frac{\pi \Jth_1'(0)}{2\ga\Jth_1(\om)}, \quad \z = \frac{t}{\ga}.
\end{equation}
\end{prop}
This proposition was proved in \cite{BL10} using the method of nonlinear steepest descent for the discrete Riemann--Hilbert problem associated with the orthogonal polynomials.
Since the weight of the orthogonal polynomials $P^{\rm HT}_j(\ell)$ differs from that of $P^{\rm DW}_j(\ell)$ only by a shift, we can expect a very similar analysis for the orthogonal polynomials $P^{\rm HT}_j(\ell)$. Indeed, making minor modifications to the analysis of \cite{BL10}, \cite[Chapter 7]{BL14}, we can prove the following theorem
\begin{theo}\label{af_HT_hk}
As $k\to\infty$, the normalizing constants $h_k^{\rm HT}$ satisfy
\begin{equation*}\label{af:10ht}
\frac{h_k^{\rm HT}}{(k!)^2} = G^{2k+1} \frac{\Jth_3((k+1)\om)}{\Jth_3(k\om)} \left(1+\ocal(k^{-2})\right),
\end{equation*}
where $G$, $\om$, and $\z$ are defined in \eqref{af:11}.
\end{theo}
Combining Proposition \ref{af_DW_hk} and Theorem \ref{af_HT_hk} with equation \eqref{eq:Zn_factor} proves Theorem \ref{theorem:AF}, up to the statement that the constant $C$ does not depend of $\ga$. This fact follows from the fact that both the determinants $\tau_n^{\rm DW}$ and $\tau_n^{\rm HT}$ satisfy the Toda equation \eqref{toda}, and is proven exactly as in \cite[Section 7.7]{BL14}.

To prove Theorem \ref{af_HT_hk} we describe the minor changes which must be made in the Riemann--Hilbert analysis of \cite[Chapter 7]{BL14} to deal with the orthogonal polynomials $P^{\rm HT}_j(\ell)$ instead of $P^{\rm DW}_j(\ell)$. The key difference is that the weight of orthogonality is shifted. After rescaling the polynomials as in \cite[Equation (7.1.9)]{BL14}, the lattice supporting the measure of orthogonality given in \cite[Equation (7.1.10)]{BL14} is 
\begin{equation*}
L_n = \left\{ x_k = \frac{\ga}{n}(2k-1) : k\in \Z\right\}.
\end{equation*}
Then the function $\Pi(z)$ defined in \cite[Equation (7.5.7)]{BL14} used to interpolate the poles in the Riemann--Hilbert problem should be changed to
\begin{equation*}
\Pi(z) = \frac{2\ga}{n\pi}\cos\left(\frac{n\pi z}{2\ga}\right),
\end{equation*}
and the interpolating matrices $D_\pm^u$ and $D_\pm^l$ defined in \cite[Equations (7.5.7) and (7.5.10)]{BL14} should be
\begin{equation*}
D_\pm^u(z) = \begin{pmatrix} 1 & -\frac{\De_n w_n(z)}{\Pi(z)} e^{\pm i\pi(\frac{nz}{2\ga}+\frac{1}{2})} \\ 0 & 1 \end{pmatrix}, \quad
D_\pm^l(z) = \begin{pmatrix} \Pi(z)^{-1} & 0 \\ -\frac{1}{\De_n w_n(z)} e^{\pm i\pi(\frac{nz}{2\ga}+\frac{1}{2})} & 1 \end{pmatrix}.
\end{equation*}
Also in the second transformation of the RHP which defines ${\bf S}_n(z)$ in terms of ${\bf T}_n(z)$, exponential factors of the form $e^{\pm in\pi z/(2\ga)}$ should be replaced with $e^{\pm i\pi(\frac{nz}{2\ga}+\frac{1}{2})}$. The result is that in the jump matrices for ${\bf S}_n(z)$ given in \cite[equation (7.5.24)]{BL14}, all exponential factors of the form $e^{\pm in\pi z/\ga}$ are replaced with $-e^{\pm in\pi z/\ga}$. This plays little role for the jump matrices which converge to the identity matrix as $n\to\infty$, but the diagonal terms of the jump matrix on the interval $[\al',\be']$ are changed by sign. Then the phase $\Om_n$ defined in \cite[Equation (7.5.26)]{BL14} should be
\begin{equation*}
\Om_n = n\pi(1+\z).
\end{equation*}
Analysis throughout the rest of the chapter is nearly identical, with $\Om_n$ shifted by $\pi$. The solution to the model RHP presented in \cite[equation (7.5.56)]{BL14} involves $\Jth_3$ with an argument involving $\Om_n/2$. When $\Om_n/2$ is shifted by $\pi/2$, it is equivalent to replacing $\Jth_3$ with $\Jth_4$. It follows then that the result of \cite[Proposition 7.5.1]{BL14} holds with $\Jth_3$ replacing $\Jth_4$, and thus \cite[equation (7.6.76)]{BL14} also holds with $\Jth_3$ replacing $\Jth_4$. The rest of the analysis in \cite[Chapter 7]{BL14} is identical, and we arrive at Theorem \ref{af_HT_hk}, which is identical to Proposition \ref{af_DW_hk} but with $\Jth_3$ replacing $\Jth_4$.

\section{Ferroelectric phase}\label{f-phase}
Now consider the ferroelectric phase $b>a+c$, where the weights $a$, $b$, and $c$ are parametrized by \eqref{pf4}.
In this case the two relevant systems of monic orthogonal polynomials are defined via the orthogonality conditions
\begin{equation}\label{f:7}
\begin{aligned}
 \sum_{\ell=0}^\infty P^{\rm DW}_j(\ell)P^{\rm DW}_k(\ell)[e^{-2\ell (t-\ga)}-e^{-2\ell(t+\ga)}]&=h^{\rm DW}_k\de_{jk}, \\
  \sum_{\ell=0}^\infty P^{\rm HT}_j(\ell)P^{\rm HT}_k(\ell)[e^{-2\ell (t-\ga)}+e^{-2\ell(t+\ga)-2\ga}]&=h^{\rm HT}_k\de_{jk}.
  \end{aligned}
\end{equation}
As shown in \cite{BL09-1}, \cite[Chapter 8]{BL14}, the DWBC partition function is given as
\begin{equation*}
Z_n^{\rm DW} = (2ab)^{n^2} \prod_{k=0}^{n-1} \frac{h^{\rm DW}_k}{(k!)^2}.
\end{equation*}
Using \eqref{in7}, \eqref{lt3a}, and \eqref{in9}, we also find that $\widetilde{Z}_n^{\rm HT}$ is given as
\begin{equation*}
\widetilde{Z}_n^{\rm HT} = (2ab)^{n^2} \prod_{k=0}^{n-1} \frac{h^{\rm HT}_k}{(k!)^2},
\end{equation*}
thus \eqref{in8} becomes
\begin{equation}\label{eq:F_factored}
Z_{2n}^{\rm HT} = (2ab)^{2n^2} \prod_{k=0}^{n-1} \frac{h^{\rm DW}_k}{(k!)^2}\prod_{k=0}^{n-1} \frac{h^{\rm HT}_k}{(k!)^2}.
\end{equation}

In the previous work \cite{BL09-1}, see also \cite[Chapter 8]{BL14}, the constants $h_k^{\rm DW}$ were compared asymptotically with the analogous constants associated with the classical Meixner polynomials, defined for example in \cite{KLS}. The classical Meixner polynomials depend on two parameters denoted $q$ and $\be$ in \cite{KLS}, and we are interested in the specialization $\be=0$. The shifted and monic version of these polynomials satisfy the orthogonality condition
\begin{equation*} \label{meix4}
\sum_{\ell=0}^\infty P_j^{\rm Q}(\ell)P_k^{\rm Q}(\ell)q^j=h_k^{\rm Q}\de_{jk},
\qquad h_k^{\rm Q}= \frac{(k!)^2q^{k+1}}{(1-q)^{2k+1}},
\end{equation*} 
see \cite[Section 8.2]{BL14} for a discussion.
Both of the weights in \eqref{f:7} are very close to the Meixner weight with $q=e^{-2(t-\ga)}$, in that they are both of the form $q^\ell\pm\tilde{q}^\ell$ where $\tilde{q}<q$. It is shown in \cite{BL09-1} that the smaller term $\tilde{q}^\ell$ has vey little effect on the constants $h_k^{\rm DW}$ as $k\to\infty$. The following proposition repeats the result of \cite[Theorem 1.1]{BL09-1} for $h_k^{\rm DW}$ and extends it to $h_k^{\rm HT}$. The proof of the result for $h_k^{\rm HT}$ is identical to the proof for $h_k^{\rm DW}$ and we refer the reader to \cite[Section 8.4]{BL14}. 
\begin{prop}
For any $\ep>0$, as $k\to\infty$ the normalizing constants $h_k^{\rm DW}$ and $h_k^{\rm HT}$ defined in \eqref{f:7} satisfy
\begin{equation*} \label{meix5}
h_k^{\rm DW} =h_k^{\rm HT}\left(1+\ocal\left(e^{-k^{1-\ep}}\right)\right) =h_k^{\rm Q}\left(1+\ocal\left(e^{-k^{1-\ep}}\right)\right),
\end{equation*}
where 
\begin{equation*} 
q=e^{-2(t-\ga)}.
\end{equation*}
\end{prop}
Applying this proposition to the formula \eqref{eq:F_factored}, we obtain
\begin{equation*}\label{f:9}
\begin{aligned}
Z^{\rm HT}_{2n}&=(2ab)^{2n^2}\prod_{k=0}^{n-1} \frac{q^{2k+2}}{(1-q)^{4k+2}}\left(1+\ocal\left(e^{-n^{1-\ep}}\right)\right) \\
&=(2ab)^{2n^2} \frac{q^{n^2+n}}{(1-q)^{2n^2}}\left(1+\ocal\left(e^{-n^{1-\ep}}\right)\right) \\
&=C \sinh(t+\ga)^{2n^2}e^{-2n(t-\ga)}\left(1+\ocal\left(e^{-n^{1-\ep}}\right)\right),
\end{aligned}
\end{equation*}
where
\begin{equation}\label{f:10}
C=\prod_{k=0}^{\infty} \frac{h_k^{\rm DW}}{h_k^{\rm Q}}\frac{h_k^{\rm HT}}{h_k^{\rm Q}}.
\end{equation}
The value of $ \prod_{k=0}^{\infty} \frac{h_k^{\rm DW}}{h_k^{\rm Q}}$ was found in \cite{BL09-1} to be
\begin{equation}\label{const_DW}
\prod_{k=0}^{\infty} \frac{h_k^{\rm DW}}{h_k^{\rm Q}} = 1-e^{-4\ga}.
\end{equation}
In the proof of this fact, see  \cite[section 8.5]{BL14}, we consider the regime $\ga>0$ is fixed and $t\to+\infty$. In this regime it is shown that the primary contribution to the product comes from the first factor, 
\begin{equation*}
\frac{h_0^{\rm DW}}{h_0^{\rm Q}}= 1-e^{-4\ga}+\ocal(e^{-2t}).
\end{equation*}
Then the fact that the Hankel determinant $\tau_n^{\rm DW}$ satisfies the Toda equation \eqref{toda} is used to show that the constant is in fact simply $1-e^{-4\ga}$. 

For the product $\prod_{k=0}^{\infty} \frac{h_k^{\rm HT}}{h_k^{\rm Q}}$ we can follow the same procedure. The only difference is that the first factor in the product is
\begin{equation*}
\frac{h_0^{\rm HT}}{h_0^{\rm Q}}= 1+e^{-4\ga}+\ocal(e^{-2t}).
\end{equation*}
All other calculations are identical to those in \cite[section 8.5]{BL14}, and the result is that 
\begin{equation}\label{const_HT}
\prod_{k=0}^{\infty} \frac{h_k^{\rm HT}}{h_k^{\rm Q}} = 1+e^{-4\ga}.
\end{equation}
Combining \eqref{f:10}, \eqref{const_DW}, and \eqref{const_HT} proves the formula \eqref{F_constants} for $C$. This completes the proof of Theorem \ref{theorem:F}.

%We now rescale, setting
%\begin{equation}
%\De_n = \frac{\ga}{n}, \qquad x=\ell \De_n, \qquad w_n(x) = e^{-n(|x|-\z x)}, \qquad \z=\frac{t}{\ga},
%\end{equation}
%and 
%\begin{equation}
%P_{nk}(x)=\De_n^k P_k\left(\frac{x}{\De_n}\right), \quad L_n=\left\{ x = \frac{\ga(2k+1)}{n} \ \bigg| \ k\in \Z\right\}.
%\end{equation}
%We then have that the polynomials $P_{nk}(x)$ satisfy the orthogonality condition
%\begin{equation}
%\sum_{x\in L_n} P_{nk}(x)P_{nj}(x)\De_n = h_{nk}\de_{jk}, \qquad h_{nk}=h_k \De_n^{2k+1},
%\end{equation}
%thus 
%\begin{equation}
%\tau_n^{\rm HT}=(2\sinh(2\ga))^n\prod_{k=0}^{n-1} \De_n^{-(2k+1)} h_{nk}=(2\sinh(2\ga))^n\left(\frac{n}{\ga}\right)^{n^2}\prod_{k=0}^{n-1} h_{nk}.
%\end{equation}

\end{document}